\documentclass[10pt]{article}

\usepackage{fullpage}
\usepackage{bm} 
\usepackage{url}
\usepackage{amsmath}
\usepackage{amsfonts}
\usepackage{amsthm}
\usepackage{algorithmic}
\usepackage{graphicx}
\usepackage{textcomp}
\usepackage{xcolor}
\usepackage{enumerate}
\usepackage{booktabs}

\def\BibTeX{{\rm B\kern-.05em{\sc i\kern-.025em b}\kern-.08em
    T\kern-.1667em\lower.7ex\hbox{E}\kern-.125emX}}

\usepackage{hyperref}
\usepackage[colorinlistoftodos]{todonotes}

\usepackage[ruled, noend]{algorithm2e}

\usepackage{tikz}
\usepackage{tikz-3dplot}

\definecolor{lightred}{rgb}{1,0.7,0.7}
\definecolor{lightgreen}{rgb}{0.7,1,0.7}
\definecolor{lightblue}{rgb}{0.7,0.7,1}

\definecolor{lightblue}{rgb}{0.55,0.72,0.97}
\definecolor{lightyellow}{rgb}{1,1,0.4} 
\definecolor{lightorange}{rgb}{1,0.7,0.5}
\definecolor{light}{gray}{0.85}
\definecolor{heavy}{gray}{0.35}
\definecolor{goodgreen}{rgb}{0.1, 0.5, 0.1}
\definecolor{apricot}{rgb}{0.98, 0.81, 0.69}
\definecolor{MidnightBlue}{rgb}{0.1, 0.1, 0.44}
\definecolor{Purple}{HTML}{911146}

\colorlet{LightViolet}{violet!40}
\colorlet{LightRed}{red!40}
\colorlet{LightOrange}{orange!40}
\colorlet{LightGreen}{green!40}
\colorlet{LightBlue}{blue!40}
\colorlet{DarkGreen}{green!50!black}
\colorlet{DarkRed}{red!70!black}
\colorlet{DarkCyan}{red!70!black}
\colorlet{DarkBlue}{blue!80!black}
\definecolor{DarkOrange}{rgb}{1.0, 0.49, 0.0}
\definecolor{Airforceblue}{rgb}{0.36, 0.54, 0.66}

\usepackage{microtype}

\newcommand{\calD}{\mathcal D}

\newtheorem{thm}{Theorem}[section]
\newtheorem{definition}[thm]{Definition}

\newtheorem{ex}[thm]{Example}

\newcommand{\nop}[1]{}

\definecolor{mMediumBrown}{HTML}{C87A2F}  

\newcommand{\bigO}[1]{\mathcal{O}(#1)}
\newcommand{\fw}{\mathsf{w}}
\newcommand{\dfw}{\delta}
\newcommand{\eps}{\epsilon}   

\newcommand{\ivme}{\text{IVM}$^{\eps}$\xspace}

\begin{document}

\title{Recent Increments in Incremental View Maintenance}
\author{Dan Olteanu\\ dan.olteanu@uzh.ch \\ University of Zurich}
\date{}

\maketitle

\begin{abstract}
We overview recent progress on the longstanding problem of incremental view maintenance (IVM), with a focus on the fine-grained complexity and optimality of IVM for classes of conjunctive queries. This  theoretical progress guided the development of IVM engines that reported practical benefits in academic papers and industrial settings.
When taken in isolation, each of the reported advancements is but a small increment. Yet when taken together, they may well pave the way to a deeper understanding of the IVM problem.

This paper accompanies the invited Gems of PODS 2024 talk with the same title. Some of the works highlighted in this paper are based on prior or on-going collaborations with: Ahmet Kara, Milos Nikolic, and Haozhe Zhang in the F-IVM project; and Mahmoud Abo Khamis, Niko G\"obel, Hung Ngo, and Dan Suciu at RelationalAI.
\end{abstract}

\section{Incremental View Maintenance}
\label{sec:problem-setting}

Incremental View Maintenance (IVM) is a fundamental problem: Given a query and a database, the goal is to maintain the output of the query under updates to the input database, e.g.,~\cite{DBT:VLDBJ:2014, Idris:dynamic:SIGMOD:2017, DattaKMSZ18, Kara:TODS:2020, TaoY22, FIVM:VLDBJ:2023, DBLP:journals/pvldb/WangHDY23, BudiuCMRT23}. The updates can be inserts and deletes\footnote{The term IVM is a misnomer, as the word {\em incremental} wrongly suggests that only inserts are supported, with no support for deletes, which are denoted by the word {\em decremental}. An alternative term coined for this problem is {\em fully dynamic computation}, where the word {\em fully} emphasizes that both inserts and deletes are supported.} of tuples into and respectively from the input database.
The motivation for this problem is the observation that small changes to the input database may beget small changes to the query output. It is therefore deemed inefficient to recompute the query output on the updated database from scratch. Instead, it may be more efficient to propagate the change from the input database to the query output.

    \begin{figure}
        \centering
    \begin{tikzpicture}
    \node at (1.0,0.25) {\small Query};
    \node at (1.0,-0.25) {{\Large\bf +}};
    \node at (1.0,-0.75) {\small Database};

    \node at (4.5,-0.3){
        \begin{tikzpicture}
            \tikzstyle{background}=[rectangle, rounded corners=2mm]
            \fill[rounded corners=2mm, fill=blue!20,inner sep= 2.5mm] (0,0) rectangle (1.7,1.5) node[pos=.5, text width=1.2cm, align=center] {\small Data\\ Structure};
        \end{tikzpicture}
    };
    
    \node at (1.6, -0.3) (anchor1''){};
    \node at (3.8, -0.3) (anchor2''){};
    \node at (2.7, 0.1) {\small preprocess};
    
    \draw[->,color = blue,minimum height=4cm,minimum width=7cm, line width = 0.07cm](anchor1'')-- (anchor2'');
    
    \node at (2.7, -0.7) {\em\small \color{mMediumBrown} preprocessing};
    \node at (2.7, -1.1) {\em\small \color{mMediumBrown} time};

    \node at (6.1, -1.0) {\small enumerate};
    \draw[dashed,->,color = blue,minimum height=4cm,minimum width=7cm, line width = 0.07cm](5.4,-0.6)-- (6.95,-0.6);
    \node at (7.83,-1.8cm){
        \begin{tikzpicture}
            \tikzstyle{background}=[rectangle, rounded corners=2mm]
            \fill[rounded corners=2mm, fill=blue!20,inner sep= 2.5mm] (0,0) rectangle (1.75,2.2) node[pos=.5, align=center] {};
            \node at (0.85, 1.9) {\small tuple $\#1$};
            \node at (0.85, 1.45) {\small tuple $\#2$};
            \node at (0.85, 1.05) {$\dots$};
            \node at (0.85, 0.7) {\small tuple $\#n$};
            \node at (0.85, 0.3) {\small EOS};

            \node at (0.85, -0.5) {\small Output};
        \end{tikzpicture}
    };
    \node[rotate=90, text width=2cm, align=center] at (9.25, -1.5cm) {\em\small \color{mMediumBrown} enumeration delay};

    \node[text width=2cm, align=center] at (6.2, 0.5) {\small access\\request};
    \draw[<-,color = blue,minimum height=4cm,minimum width=7cm, line width = 0.07cm](5.4,0.0)-- (7.0,0.0);
    \node at (7.8,0.1cm){
        \small User
    };

    \node at (0.97, -3.3) (anchor1''''){};
    \node at (0.97, -1.1) (anchor2''''){};
    \node at (0.97, -3.7) {update};
    
    \draw[->,color = blue,minimum height=4cm,minimum width=7cm, line width = 0.07cm](anchor1'''')-- (anchor2'''');
    
    \node at (1.5, -3.7) (anchor3){};
    \node at (4.5, -3.7) (anchor4){};
    \node at (4.4, -3.85) (anchor42){};
    \node at (4.4, -1.1) (anchor5){};
    
    \draw[color = blue,minimum height=4cm,minimum width=7cm, line width = 0.07cm](anchor3)-- (anchor4);
    
    \draw[->,color = blue,minimum height=4cm,minimum width=7cm, line width = 0.07cm](anchor42)-- (anchor5);
    
    \node at (2.8, -3.4) {\small maintain};

    \node[rotate=90] at (0.65, -2.2) {\small maintain};

        \node at (2.8, -4.0) {\em\small \color{mMediumBrown} update time};

        \node[rotate=90] at (1.3, -2.2) {\em\small \color{mMediumBrown} update time};
    \end{tikzpicture}

    \caption{The Incremental View Maintenance problem: Given a query and a database, we construct a data structure  that encodes the query output. An update triggers changes to both the database and the data structure. On a user access request, the query output is enumerated one tuple at a time.}
    \label{fig:IVM-problem}
    \end{figure}
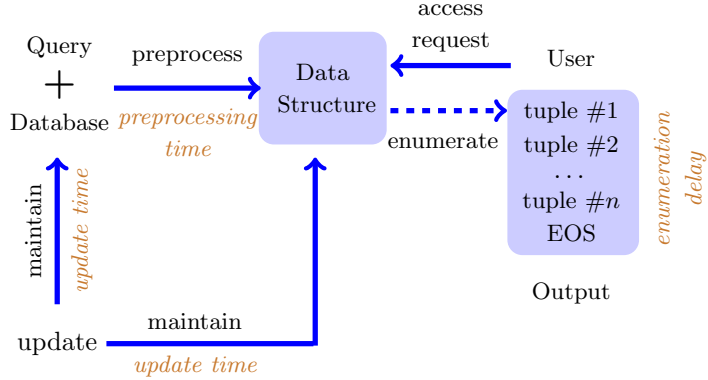

Fig.~\ref{fig:IVM-problem} gives a refinement of the IVM problem, where its time complexity is decomposed into {\em preprocessing time}, {\em update time}, and {\em enumeration delay}. The preprocessing time is the time to compute a  data structure that encodes the query output. The update time is the time to update the database and the data structure on a single-tuple insert or delete. The enumeration delay is the maximum of three times~\cite{DurandFO07}: the time between the start of the enumeration and outputting the first tuple, the time between outputting any two consecutive tuples, and the time between outputting the last tuple and the end of the enumeration\footnote{This is sometimes referred to as {\em full enumeration}. A variant is {\em delta enumeration}, where we only enumerate the tuples in the change to the query output.}.

This paper overviews a recent trend in the study of the IVM problem that aims to understand its fine-grained complexity and its optimality. It is concerned with questions such as: 
\begin{itemize}
    \item For which IVM workload can we achieve constant update time and enumeration delay?
    \item More generally, when can we achieve worst-case optimal update time and enumeration delay? 
    \item What is the trade-off between the preprocessing time, the update time, and the enumeration delay?
\end{itemize}
These recent IVM efforts aim to mirror the development of worst-case optimal join algorithms in the static setting (where the input database does not change)~\cite{Yannakakis:VLDB:81,LeapFrogTrieJoin2014,Ngo:JACM:18}, albeit in the more general dynamic setting. 
Existing answers to these questions rely on classical IVM techniques, such as delta queries and materialized views, and on newer IVM techniques, such as query factorization and heavy/light data partitioning. 

A remarkable feature of this new development is that it can achieve lower complexity and faster runtime in practice when compared with the classical IVM approaches based on delta queries and materialized views. It has also been implemented to a varying extent in open-source prototypes, e.g., \textsf{DBToaster}~\cite{DBT:VLDBJ:2014}, \textsf{Dynamic Yannakakis}~\cite{Idris:dynamic:SIGMOD:2017,DBLP:journals/vldb/IdrisUVVL20},  \textsf{F-IVM}~\cite{DBLP:conf/sigmod/NikolicO18,FIVM:VLDBJ:2023}, and Crown~\cite{DBLP:journals/pvldb/WangHDY23}, and in commercial systems, e.g., RelationalAI (\url{https://relational.ai}).

This development poses formidable challenges but also exciting opportunities. 
Beyond notable yet limited explorations, understanding the optimality of the IVM problem for the entire language of conjunctive queries remains open. Further progress would require to bridge database theory to recent developments in fine-grained complexity for dynamic algorithms and to contribute with novel lower bound proof techniques and efficient algorithms.

We describe key IVM techniques in Sec.~\ref{sec:ivm-techniques}. We use them to achieve the best possible maintenance for a variety of IVM workloads in Sec.~\ref{sec:q} and \ref{sec:optimality}. We point to some relevant work we left out in Sec.~\ref{sec:further}.

\section{Preliminaries}

A schema $\mathcal{S}$ is a tuple of variables, which we also see as a set. Let $\mathsf{Dom}(X)$ denote the domain of a variable $X$. A tuple $t$ over schema $\mathcal{S}$ has the domain $\mathsf{Dom}(\mathcal{S}) = \prod_{X \in \mathcal{S}}{\mathsf{Dom}(X)}$. The empty tuple $()$ is the tuple over the empty schema. 
Given a ring $(\mathbf{D}, +,*, {\bf 0}, {\bf 1})$ and a schema $\mathcal{S}$, a relation $R$ over  $\mathcal{S}$ and $\mathbf{D}$ is a function $R: \mathsf{Dom}(\mathcal{S}) \to \mathbf{D}$ mapping tuples over schema $\mathcal{S}$ to values in $\mathbf{D}$ such that $R(t) \neq {\bf 0}$ for finitely many tuples $t$. The tuple $t$ is called a {\em key}, while its mapping $R(t)$ is the {\em payload} of $t$ in $R$. The size $|R|$ of $R$ is the number of tuples in $R$ mapped to non-zero payloads. A database $\mathcal{D}$ is a set of relations over the same ring. Its size $|\mathcal{D}|$ is the sum of the sizes of its relations. 
Each relation $R$ is implemented as a hash map with key-payload entries. One can look up, insert, and delete entries from this map in amortized constant time, and enumerate its entries with constant delay. For a subset $\mathcal{S}$ of the schema of $R$, we need at times an index that can enumerate with constant delay all tuples from $R$ that agree on a given projection $t\in\pi_{\mathcal S} R$ and insert and delete index entries in amortized constant time~\cite{FIVM:VLDBJ:2023}.

We consider the language of SQL queries with natural joins and group-by aggregates. This includes conjunctive queries, where aggregates are projections. We use the following query syntax:
\begin{align*}
  Q(X_1,\ldots,X_f) = \sum_{X_{f+1}} \ldots \sum_{X_{m}} \prod_{i \in [n]} R_i(\mathcal{S}_i),  
\end{align*}
where $X_1,\ldots,X_f$ are the free (group-by) variables, $X_{f+1},\ldots,X_m$ are the bound variables, $R_i(\mathcal{S}_i)$ are the atoms of $Q$, $\sum_{X_{f+1}}$ is the aggregation operator that  marginalizes the bound variable $X_{f+1}$ according to the sum operation $+$ in the ring. For a variable $X$, $\mathit{atoms}(X)$ is the set of atoms $R_i(\mathcal{S}_i)$ with $X\in\mathcal{S}_i$.
We also use a union operator $\uplus$ to express updates (insert/delete) to relations. Given relations $R$ and $S$ over schema $\mathcal{S}_1$ and relation $T$ over schema $\mathcal{S}_2$, a variable $X \in \mathcal{S}_1$, we define the three operators as follows:

\begin{tabular}{@{~~~~}l@{~~~~~}r@{~}r@{\;}l@{}}
\multicolumn{4}{@{~}l}{\em union:}\\
& $\forall t \in \mathsf{D}_1{:}$ 
& $(R + S)(t)$ 
& $= R(t) + S(t)$ \\[4pt]
\multicolumn{4}{@{~}l}{\em join:}\\
& $\forall t \in \mathsf{D}_2{:}$ 
& $(S \cdot T)(t)$ 
& $= S(\pi_{\mathcal{S}_1}(t)) * T(\pi_{\mathcal{S}_2}(t))$ \\[4pt]
\multicolumn{4}{@{~}l}{\em aggregation:}\\
& $\forall t \in \mathsf{D}_3{:}$ 
& $(\sum_{X} R)(t)$ 
& $= \textstyle\sum \,\{\, R(t_1) * g_{X}(\pi_{\{X\}}(t_1)) \mid$ \\
& & & \hspace*{2em} $t_1 \,{\in}\, \mathsf{D}_1, t = \pi_{\mathcal{S}_1 \setminus \{X\}}(t_1) \}$ 
\end{tabular}
\\[6pt]
where $\mathsf{D}_1 = \mathsf{Dom}(\mathcal{S}_1)$, $\mathsf{D}_2 = \mathsf{Dom}(\mathcal{S}_1\cup \mathcal{S}_2)$, and $\mathsf{D}_3 = \mathsf{Dom}(\mathcal{S}_1\setminus \{X\})$, and $\pi_{\mathcal{S}}(t)$ is a tuple representing the projection of tuple $t$ on the schema $\mathcal{S}$. The function $g_X$ is a {\em lifting} function that maps $X$-values to ring elements. When marginalizing $X$,  we aggregate the values $g_X(x)$ from $\mathbf{D}$ and not the values $x$ from $\mathsf{Dom}(X)$.

There are two benefits of using relations over rings: simplicity and optimization. Simplicity is due to the uniform treatment of updates: Updates are tuples, which are mapped to positive ring values in case of inserts and negative ring values in case of deletes. 
The optimization benefit is due to the commutativity of inserts and deletes: A batch of updates can be executed in any order, their cumulative effect remains the same. This is particularly relevant for asynchronous execution of updates in a distributed setting.

The ring of integers $(\mathbb{Z}, +,*, 0, 1)$ is used in the literature to maintain the tuple multiplicities, e.g., in the DBToaster and F-IVM systems~\cite{DBT:VLDBJ:2014,FIVM:VLDBJ:2023}. A positive multiplicity is  interpreted as the number of copies of a tuple; for a tuple in the query output, this is the number of derivations from the input tuples that led to that output tuple. A negative multiplicity may appear as side effect of out-of-order updates. We assume that batches of updates are valid in the sense that they map between databases, where all (input and output) tuples have positive multiplicities.

Our data model follows prior work on $K$-relations over provenance semirings~\cite{Green:2007:ProvenanceSemirings}, generalized multiset relations~\cite{DBLP:conf/pods/Koch10}, and factors over semirings~\cite{FAQ:PODS:2016}.

\section{IVM Techniques by Example}
\label{sec:ivm-techniques}

\begin{figure*}[t]
    \centering
  \begin{tikzpicture}
       \node at(0, 0) {
  \begin{tabular}{l@{\hspace*{2em}}l@{\hspace*{2em}}l@{\hspace*{2em}}l@{\hspace*{2em}}l}
  
  \begin{tabular}{c@{\;}c|@{\;}c@{\;}}
    \multicolumn{3}{c}{$R$} \\
    \toprule
    $A$ & $B$ & \# \\\midrule
    $a_1$ & $b_1$ & $2$\\
    $a_2$ & $b_1$ & $3$\\
    \bottomrule
    \\
  \end{tabular}
  
  &
  
  \begin{tabular}{c@{\;}c|@{\;}c@{\;}}
    \multicolumn{3}{c}{$S$} \\
    \toprule
    $B$ & $C$  & \# \\\midrule
    $b_1$ & $c_1$ & $2$\\
    $b_1$ & $c_2$ & $1$\\
    \bottomrule
    \\
    \end{tabular}
  
  &
  
  \begin{tabular}{c@{\;}c|@{\;}c@{\;}}
    \multicolumn{3}{c}{$T$} \\
    \toprule
    $C$ & $A$ & \# \\\midrule
    $c_1$ & $a_1$ & $1$\\
    $c_2$ & $a_1$ & $3$\\
    $c_2$ & $a_2$ & $3$\\
    \bottomrule
  \end{tabular}

  &

  \begin{tabular}{c@{\;}c@{\;}c|@{\;}c@{\;}}
    \multicolumn{4}{c}{$R \cdot S \cdot T$} \\
    \toprule
    $A\,  $ & $B$ & $C$ &  \# \\\midrule 
    $a_1\,  $ & $b_1$  & $c_1$ & $2 \cdot 2 \cdot 1 = 4$ \\
    $a_1\,  $ & $b_1$  & $c_2$ & $2 \cdot 1 \cdot 3 = 6$ \\
    $a_2\,  $ & $b_1$  & $c_2$ & $3 \cdot 1 \cdot 3 = 9$ \\  
    \bottomrule 
  \end{tabular}

  &

  \begin{tabular}{c@{\;}|@{\;}c@{\;}}
    \multicolumn{2}{c}{$Q$} \\
    \toprule
    $\emptyset$ & \#  \\\midrule
    $(\, )$ & $4 + 6 + 9 = 19$\\
    \bottomrule
    \multicolumn{2}{c}{} \\ 
    \multicolumn{2}{c}{} \\ 
  \end{tabular}
  \end{tabular}
  };
  \end{tikzpicture}

  \begin{tikzpicture}
       \node at(0, 0) {
  \begin{tabular}{l@{\hspace*{2em}}l@{\hspace*{2em}}l@{\hspace*{2em}}l}
  \begin{tabular}{c@{\;}c|@{\;}c@{\;}}
    \multicolumn{3}{c}{$\delta R = \{(a_2,b_1) \mapsto -2\}$} \\
    \toprule
    $A$ & $B$ & {\#} \\\midrule
    $a_2$ & $b_1$ & ${\color{red}-2}$\\
    \bottomrule
  \end{tabular}
  & 
  \begin{tabular}{c@{\;}c@{\;}c|@{\;}c@{\;}}
    \multicolumn{4}{c}{$\delta (R \cdot S \cdot T)$} \\
    \toprule
    $A\,  $ & $B$ & $C$ &  \# \\\midrule 
    $a_2\,  $ & $b_1$  & $c_2$ & ${\color{red}-2 \cdot 1 \cdot 3 = -6}$ \\  
    \bottomrule 
  \end{tabular}
  &
  \begin{tabular}{c@{\;}|@{\;}c@{\;}}
    \multicolumn{2}{c}{$\delta Q$} \\
    \toprule
    $\emptyset$ & \#  \\\midrule
    $(\, )$ & ${\color{red} -6}$\\
    \bottomrule
  \end{tabular}
  &
  \begin{tabular}{c@{\;}|@{\;}c@{\;}}
    \multicolumn{2}{c}{$Q \uplus \delta Q$} \\
    \toprule
    $\emptyset$ & \#  \\\midrule
    $(\, )$ & ${\color{red} 19 + - 6 = 13}$\\
    \bottomrule
  \end{tabular}

  \end{tabular}
  };
  \end{tikzpicture}
    \caption{IVM example. Top row: Computation of the triangle query over a database with integer payloads. Bottom row: Delta views and updated query output after a delete to $R$.}
    \label{fig:IVM-triangle-example}
\end{figure*}
  
We exemplify IVM techniques using the Boolean triangle query: 
\begin{align*}
    Q = \sum_{A,B,C} R(A,B)\cdot S(B,C)\cdot T(C,A)
\end{align*}
and databases over the ring of integers $(\mathbb{Z}, +,*, 0, 1)$. 

\subsection{Delta Queries}

The delta queries represent the most widely supported IVM technique~\cite{Chirkova:Views:2012:FTD}. We explain it next. We denote the update to a relation $R$ as the {\em delta relation} $\delta R$ that maps the tuples to be inserted to $+1$ and the tuples to be deleted to $-1$. The updated relation is  the union of the old relation and the delta relation: $R := R \uplus \delta R$. 
For a materialized view $V$, the {\em delta view} $\delta V$ defines the change to its content. If $V$ is an input relation $R$, then $\delta V = \delta R$. If $V$ is defined using operators on other views, $\delta V$ is derived using the following {\em delta rules} (multiplication with the empty delta view $\emptyset$ yields $\emptyset$):
\begin{align}
  \delta (V_1 \uplus V_2) &= \delta V_1 \uplus \delta V_2 \label{eq:delta-union}\\
  \delta (V_1 \cdot V_2) &= (\delta V_1 \cdot V_2) \uplus (V_1 \cdot \delta V_2) \uplus (\delta V_1 \cdot \delta V_2) \label{eq:delta-product}\\
  \delta \left(\sum_{X} V\right) &= \sum_{X} \delta V \label{eq:delta-sum}
\end{align}

\begin{ex}\em \label{ex:triangle-delta}
 Fig.~\ref{fig:IVM-triangle-example} shows a database with three relations $R$, $S$, and $T$, the  output of their join, and the output of $Q$. Each tuple has a multiplicity, which is given for the input tuples and computed for the output tuples. For instance, the multiplicity of the tuple $(a_2,b_1,c_2)$ in the join output is the product of the multiplicities of $R(a_2,b_1)$, $S(b_1,c_2)$, and $T(c_2,a_2)$. The update $\delta R$ triggers changes to $R$, the join output and the query output. The new relation $R$ is $R \uplus \delta R$, where $(a_2,b_1)$ is now mapped to $3-2=1$. The change to the join output is computed using Eq.~\eqref{eq:delta-union} to \eqref{eq:delta-sum}, where $\delta S = \delta T = \emptyset$, $V \uplus \emptyset = V$, and $V \cdot \emptyset = \emptyset$ for any view $V$:
\begin{align*}
    \delta Q & =  \delta \left(\sum_{A,B,C} R(A,B) \cdot S(B,C) \cdot T(C,A) \right) \\
    &\overset{\eqref{eq:delta-sum}}{=} \sum_{A,B,C} \delta \left(R(A,B) \cdot S(B,C) \cdot T(C,A)\right) \\
    & \overset{\eqref{eq:delta-product},\eqref{eq:delta-union}}{=} \sum_{A,B,C} \delta R(A,B) \cdot S(B,C) \cdot T(C,A)\\
    &= \delta R(a_2,b_1) \cdot \sum_{C} S(b_1,C) \cdot T(C,a_2)
\end{align*}
Only one out of the three tuples in the join output changed due to the update to $R$.    
\end{ex}

Assume each input relation has size $O(N)$ and we process a single-tuple update to any of the relations. Then, recomputation from scratch of the query output  using a worst-case optimal join algorithm takes time $O(N^{3/2})$ (or slightly better using Strassen-like matrix multiplication). Using delta queries, the time for a single-tuple update is $O(N)$: For a single-tuple update $\delta R = \{(a,b)\mapsto m\}$ to $R$, we need to intersect the two lists of $C$-values that are paired with $a$ in $T$ and with $b$ in $S$. Yet we can obtain a lower update time.

\subsection{Materialized Views} 

Higher-order maintenance~\cite{DBT:VLDBJ:2014,FIVM:VLDBJ:2023} uses materialized views to speed up the delta queries. Such views need to be maintained as well.

\begin{ex}\em 
We continue Ex.~\ref{ex:triangle-delta} and now materialize the view: $V_{ST}(B,A) = \sum_C S(B,C)\cdot T(C,A)$. Then, on an update $\delta R$, the delta query is computed as:
\begin{align*}
    \delta Q = \sum_{A,B} \delta R(A,B) \cdot V_{ST}(B,A)
\end{align*}
On an update $\delta S$, we maintain both the view $V_{ST}$ and the query $Q$:
\begin{align*}
  \delta Q &= \sum_{A,B} R(A,B) \cdot \delta V_{ST}(B,A)\\
  \delta V_{ST}(B,A) &= \delta S(B,C)\cdot T(C,A)
\end{align*}
The time to process $\delta R$ is $O(|\delta R|)$, since for each single-tuple update in $\delta R$ we only need one lookup in $V_{ST}$. On the other hand, a single-tuple update $\delta S=\{(b,c)\mapsto m\}$ costs $O(N)$ time, since there may be $O(N)$ $A$-values paired with $c$ in $T$. 
If we create materialized views for the join of $R$ and $S$ and for the join of $R$ and $T$, then a single-tuple update to any input relation takes $O(N)$ time.
\end{ex}
Materialized views may reduce the update time in case of updates to one of the input relations by conveniently materializing one view, but it does not reduce the update time for updates to all relations. The price to pay is an increase in used storage: Whereas the delta query approach does not require additional space, the materialized view approach requires $O(N^2)$ extra storage for the view $V_{ST}$.

\subsection{Heavy/Light Data Partitioning} 
\label{sec:heavy-light}

We can further lower the update time by employing a new technique in addition to delta queries and materialized views: partitioning the input relations based on the heavy/light join keys and using an adaptive maintenance strategy that treats heavy keys differently from the light keys. In particular, we partition $R$ based on $A$ as $R = R_L \uplus R_H$, where $R_L = \{t\in R \mid |\sigma_{A=t.A}| < N^\epsilon\}$ is the light part 
and $R_H = R \setminus R_L$ is the heavy part. The parameter $\epsilon\in[0,1]$ determines the heavy/light threshold. We can infer that for all $A$-values $a$ in the light part, it holds: $|\sigma_{A=a} R_L| < N^\epsilon$. We can also infer that there are at most $N^{1-\epsilon}$ heavy values in the heavy part, since each of them is paired with at least $N^\epsilon$ $B$-values. Similarly, we partition $S = S_L \uplus S_H$ on $B$ and $T = T_L \uplus T_H$ on $C$. The query $Q$ becomes a summation of eight skew-aware queries, one query for each combination of parts of the three relations:
\begin{align*}
    Q = \sum_{A,B,C} R_U(A,B)\cdot S_V(B,C)\cdot T_W(C,A), \hspace*{1em} \text{for } U,V,W\in\{L,H\}
\end{align*}
We next consider a single-tuple update $\delta R = \{(\alpha,\beta)\mapsto m\}$; updates to $S$ and $T$ follow the same reasoning as the query is symmetric. On this update, we need to compute the delta queries:
\begin{align*}
    \delta Q_{VW} = \delta R(\alpha,\beta)\cdot \sum_{C} S_V(\beta,C)\cdot T_W(C,\alpha), \hspace*{1em} \text{ for } V,W\in\{L,H\} \label{eq:vw}
\end{align*}
To compute $\delta Q_{VW}$, we need to intersect the list of $C$-values paired with $\beta$ in $S_V$ and the list of $C$-values paired with $\alpha$ in $T_W$. The cost of the intersection depends on the heavy/light parts $V$ and $W$ and is proportional to the size of the smallest of the two lists. 
\begin{itemize}
    \item It takes $O(N^\epsilon)$ time to compute $\delta Q_{LL}$: There are at most $N^\epsilon$ $C$-values paired with $\beta$ in $S_L$, since $\beta$ is light in $S_L$. 
    \item It takes $O(N^{1-\epsilon})$ time to compute $\delta Q_{HH}$: There are at most $N^{1-\epsilon}$ heavy $C$-values in $T_H$. 
    \item It takes $O(N^{\min\{\epsilon,1-\epsilon\}})$ time to compute $\delta Q_{LH}$: There are at most $N^\epsilon$ $C$-values paired with $\beta$ in $S_L$ and at most $N^{1-\epsilon}$ heavy $C$-values in $T_H$. 
\end{itemize}
For $\delta Q_{HL}$, the above reasoning does not suffice. Instead, we materialize the view $V_{ST}(B,A) = \sum_C S_H(B,C)\cdot T_L(C,A)$ in a preprocessing step. Then, it takes $O(1)$ time to compute $\delta Q_{HL} = \delta R(\alpha,\beta)\cdot V_{ST}(\beta,\alpha)$, which requires a single lookup in $V_{ST}$. Yet we still need to maintain the materialized view  $V_{ST}$ under updates to $S$ and $T$. On a single-tuple update $\delta S = \{(\beta,\gamma)\mapsto m\}$, it takes $O(N^\epsilon)$ time to compute $\delta V_{ST}(\beta,A) = \delta S(\beta,\gamma) \cdot T_L(\gamma, A)$ since there are at most $N^\epsilon$ $A$-values paired with $\gamma$ in $T_L$. On a single-tuple update $\delta T = \{(\gamma,\alpha)\mapsto m\}$, it takes $O(N^{1-\epsilon})$ time to compute $\delta V_{ST}(B ,\alpha) = S_H(B,\gamma) \cdot \delta T(\gamma, \alpha)$ since there are at most $N^{1-\epsilon}$ heavy $B$-values in $S_H$. The single-tuple update time is thus the maximum of all the aforementioned times: $O(N^{\max\{\epsilon,1-\epsilon\}})$. Its minimum is $O(N^{1/2})$ reached for $\epsilon=1/2$. This update time is asymptotically better than the $O(N)$ update time obtained when using delta queries and materialized views. A further aspect to consider is that, after updates, the heavy/light partitioning can become invalid, since heavy (light) values become light (heavy). We therefore need to rebalance the partitions of the relations and recompute the materialized views after sufficiently many updates~\cite{DBLP:conf/icdt/KaraNNOZ19, Kara:TODS:2020, DBLP:conf/pods/0002NOZ20}. 

This adaptive approach was introduced in the IVM$^\epsilon$ framework for cyclic~\cite{DBLP:conf/icdt/KaraNNOZ19, Kara:TODS:2020} and hierarchical queries~\cite{DBLP:conf/pods/0002NOZ20,DBLP:journals/lmcs/KaraNOZ23}. It was extended to $k$-clique counting and to a parallel batch-dynamic triangle count algorithm~\cite{DBLP:conf/apocs/DhulipalaLSY21}. There is also work on a trade-off between the update time and approximation quality for triangle counting~\cite{DBLP:conf/icdt/LuT21}.

\subsection{Lower Bounds} 

A necessary technique in the IVM toolbox concerns proofs of lower bounds on the update time, enumeration delay, and preprocessing time. We next showcase the main technique used in the literature: A reduction from the Online Vector Matrix-Vector Multiplication (OuMv) problem~\cite{Henzinger:OMv:2015} to the triangle detection problem~\cite{BerkholzKS17,DBLP:journals/sigact/Williams18,DBLP:conf/icdt/KaraNNOZ19}, which is the Boolean version of our triangle count query where we are only interested in whether the count is greater than 0. In the following, we denote this Boolean query by $Q_b$.

\begin{definition}[\cite{Henzinger:OMv:2015}]\em 
    The {\em OuMv problem} is defined as follows:
    \begin{itemize}
        \item Input: Boolean matrix $\mathbf{M}\in\mathbb{B}^{n\times n}$ and $n$ pairs $(\mathbf{u}_r,\mathbf{v}_r)$ of Boolean column-vectors 
        $\mathbf{u}_r,\mathbf{v}_r\in\mathbb{B}^{n}$, for $r\in[n]$.
        \item Goal: After seeing each pair $(\mathbf{u}_r,\mathbf{v}_r)$, output $\mathbf{u}_r^{\top}\ \mathbf{M}\ \mathbf{v}_r$.    
    \end{itemize}
The {\em OuMv Conjecture} is as follows: For any $\gamma > 0$, there is no algorithm that solves OuMv in time $O(n^{3-\gamma})$. 
\end{definition}

We next show the following lower bound:

\begin{thm}[\cite{BerkholzKS17,DBLP:conf/icdt/KaraNNOZ19}]
    Let $\mathcal{D}$ be any database of size $N$. For any $\gamma > 0$, there is no algorithm that incrementally maintains $Q_b$ with update time $O(N^{1/2-\gamma})$ and enumeration delay  $O(N^{1-\gamma})$ unless the OuMv Conjecture fails.
\end{thm}

\begin{proof}[Proof idea]
    Assume there exists Algorithm $\mathcal{A}$ that can maintain the triangle detection query $Q_b$ with (amortized)  update time $O(N^{1/2-\gamma})$ and enumeration delay  $O(N^{1-\gamma})$ for some $\gamma > 0$. We design an algorithm $\mathcal{B}$ that uses $\mathcal{A}$ to solve OuMv in subcubic time in $n$. However, this contradicts the OuMv conjecture. 

\paragraph{Construction} Algorithm $\mathcal{B}$ is as follows:
\begin{enumerate}
\item Construct relation $S$ to encode the matrix $\mathbf{M}$: $S(i,j)=\mathbf{M}[i,j]$ 
\item In each round $r \in [n]$:
\begin{enumerate}
\item Construct relation $R$ to encode the vector $\mathbf{u}_r$: $R(a,i) = \mathbf{u}_r[i]$, for some constant value $a$ and $\forall i\in[n]$.
\item Construct relation $T$ to encode the vector $\mathbf{v}_r$: $T(j,a) = \mathbf{v}_r[j]$ for the same constant $a$ and $\forall j\in[n]$.

\item Since by construction $\mathbf{u}_r^{\top} \mathbf{M} \mathbf{v}_r = Q_b$, we check whether $Q_b=1$ using Algorithm $\mathcal{A}$.
\end{enumerate}
\end{enumerate} 

We next exemplify this construction. Each tuple in $R$ records the column of a value 1 in $\mathbf{u}_r^{\top}$. Each tuple in $S$ records the (row, column) pair of each value 1 in $\mathbf{M}$. Also, each tuple in $T$ records the row of a value 1 in $\mathbf{v}_r$. The value $\mathbf{u}_r^{\top} \mathbf{M} \mathbf{v}_r$ is the same as $Q_b$.

\begin{center}
  \begin{tikzpicture}
    \node (myNode) at (-3,1) {$\mathbf{u}^{\top}$};
    \node (myNode) at (1,1) {$\mathbf{M}$};
    \node (myNode) at (4,1) {$\mathbf{v}$};
    \node (myNode) at (7,1) {$\mathbf{u}^{\top}\mathbf{M}\mathbf{v}$};

    \node (myNode) at (-3,0) {
    \setlength\arrayrulewidth{1pt}
    \begin{tabular}{|c|c|c|}
    \hline
    $0$ &  {\color{red}$1$} & $0$\\
    \hline
    \end{tabular}
    };
    
    \node (myNode) at (1,0) {
    \setlength\arrayrulewidth{1pt}
    \begin{tabular}{|c|c|c|}
    \hline
     $0$ & $1$ & $0$\\
    \hline
    {\color{red}$1$} & $0$ & $0$ \\
    \hline
    $0$ & $0$ & $1$\\
    \hline
    \end{tabular}
    };
    
    \node (myNode) at (4,0) {
    \setlength\arrayrulewidth{1pt}
    \begin{tabular}{|c|}
    \hline
    {\color{red}$1$} \\
    \hline
    $0$ \\
    \hline
    $0$ \\
    \hline
    \end{tabular}
    };
    
    \node (myNode) at (7,0) {
    \setlength\arrayrulewidth{1pt}
    \begin{tabular}{|c|}
    \hline
    $1$\\
    \hline
    \end{tabular}
    };

    \node (myNode) at (-3.2,-2.5) {
     \begin{tabular}{cc|c}
        \multicolumn{3}{c}{$R$} \\
        \toprule
        $A$ & $B$ &  val\\\midrule
        $a$ &{\color{red}$2$} & $1$\\
        \bottomrule
        \\
        \\
      \end{tabular}
    };

    \node (myNode) at (1,-2.5) {
     \begin{tabular}{cc|c}
        \multicolumn{3}{c}{$S$} \\
        \toprule
        $B$ & $C$ &  val \\\midrule
        {\color{red}$2$} & {\color{red}$1$} & $1$\\
        $1$ & $2$ & $1$\\
        $3$ & $3$ & $1$\\
        \bottomrule
      \end{tabular}
    };
    
    \node (myNode) at (4,-2.5) {
     \begin{tabular}{cc|c}
        \multicolumn{3}{c}{$T$} \\
        \toprule
        $C$ & $A$ &  val \\\midrule
       {\color{red}$1$} & $a$ & $1$\\
            \bottomrule
            \\
            \\
      \end{tabular}
    };
    
    \node (myNode) at (7,-2.1) {
     \begin{tabular}{c|c} 
        \multicolumn{2}{c}{$Q_b$} \\
        \toprule
        $\emptyset$  &  val \\\midrule
        $(\, )$ & $1$\\
            \bottomrule
      \end{tabular}
    };
    
    \end{tikzpicture}
\end{center}

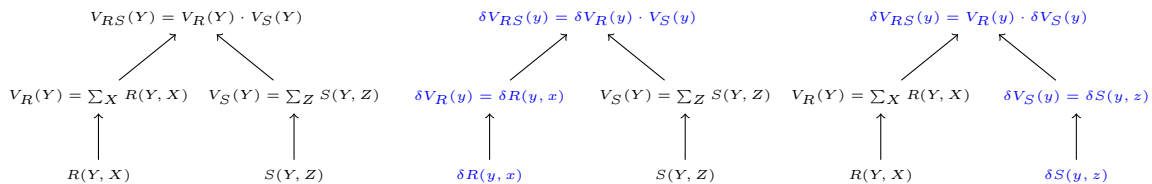
\begin{figure*}[t]
    \centering
   \begin{tikzpicture}[xscale=0.65, yscale=0.7]
    \begin{scope}
      \node at (-2, 0) (vrs) {\tiny $V_{RS}(Y) = V_R(Y)\cdot V_S(Y)$};
      \node at (-4, -1.5) (vr) {\tiny $V_R(Y) = \sum_X R(Y,X)$} edge[->] (vrs);
      \node at (0, -1.5) (vs) {\tiny $V_S(Y) = \sum_Z S(Y,Z)$} edge[->] (vrs);
      \node at (-4, -3.0) (r) {\tiny $R(Y,X)$} edge[->] (vr);
      \node at (0, -3.0) (s) {\tiny $S(Y,Z)$} edge[->] (vs);
  \end{scope}
    
  \begin{scope}
      \node at (6, 0) (vrs) {\tiny \color{blue} $\delta V_{RS}(y) = \delta V_R(y)\cdot V_S(y)$};
      \node at (4, -1.5) (vr) {\tiny \color{blue}$\delta V_R(y) = \delta R(y,x)$} edge[->] (vrs);
      \node at (8, -1.5) (vs) {\tiny $V_S(Y) = \sum_Z S(Y,Z)$} edge[->] (vrs);
      \node at (4, -3.0) (r) {\tiny\color{blue} $\delta R(y,x)$} edge[->] (vr);
      \node at (8, -3.0) (s) {\tiny $S(Y,Z)$} edge[->] (vs);
  \end{scope}

  \begin{scope}
      \node at (14, 0) (vrs) {\tiny\color{blue} $\delta V_{RS}(y) =  V_R(y)\cdot \delta V_S(y)$};
      \node at (12, -1.5) (vr) {\tiny $V_R(Y) = \sum_X R(Y,X)$} edge[->] (vrs);
      \node at (16, -1.5) (vs) {\tiny\color{blue} $\delta V_S(y) = \delta S(y,z)$} edge[->] (vrs);
      \node at (16, -3.0) (r) {\tiny\color{blue} $\delta S(y,z)$} edge[->] (vs);
      \node at (12, -3.0) (r) {\tiny $R(Y,X)$} edge[->] (vr);
  \end{scope}
    \end{tikzpicture}
    \caption{Delta view trees for the maintenance of the q-hierarchical query $Q(Y,X,Z) = R(Y,X) \cdot S(Y,Z)$.}
    \label{fig:q-hier-view-trees}
\end{figure*}

\paragraph{Complexity} Algorithm $\mathcal{B}$ is a sequence of inserts and deletes, which are executed using $\mathcal{A}$.  Overall, $\mathcal{B}$ constructs a database of size $N=O(n^2)$.
Step (1) needs less than $n^2$ inserts and an overall update time of $O(n^2 \cdot (n^2)^{\frac{1}{2}-\gamma}) = O(n^{3 - 2 \gamma})$.
Steps (2a) and (2b) need to first delete all tuples in $R$ and $T$  and then insert new tuples corresponding to the new vectors $\mathbf{u}_r[i]$ and $\mathbf{v}_r[j]$. This takes less than $4n$ updates and an overall update time of $O(4n \cdot (n^2)^{\frac{1}{2}-\gamma}) = O(n^{2 - 2 \gamma})$.
Step (2c) checks whether $Q_b=1$, which is equivalent to checking $\mathbf{u}_r^{\top}\mathbf{M}\mathbf{v}_r=1$. This takes time for one enumeration delay, so $O((n^2)^{1-\gamma}) = O(n^{2 - 2 \gamma})$.
The entire Step (2) needs $O(n(n^{2 - 2 \gamma}+n^{2 - 2 \gamma}))= O(n^{3 - 2 \gamma})$. The overall time for Algorithm $\mathcal{B}$  is thus $O(n^{3 - 2 \gamma}+n^{3 - 2 \gamma}) = O(n^{3 - 2 \gamma})$. This is subcubic in $n$ and contradicts the OuMv conjecture.
 \end{proof}

If we put together the $O(N^{1/2})$ update time and the $O(1)$ enumeration delay with the above lower bound, we conclude that the adaptive IVM approach is worst-case optimal for the triangle count query, conditioned on the OuMv conjecture.

\section{The Best Possible Maintenance}
\label{sec:q}

Notable progress has been made recently on understanding the best possible maintenance scenario: constant update time and constant enumeration delay. In this section, we overview recent works that take the structure of the queries, the database, or the update sequence into account to achieve this best possible maintenance.

\subsection{The Q-hierarchical Queries}

A seminal result is the precise syntactic characterization of all conjunctive queries without self-joins that admit the best possible maintenance for any database and update sequence:

\begin{thm}[\cite{BerkholzKS17}]\label{th:q-hierarchical}
    Let $Q$ be a self-join free conjunctive query and $\mathcal{D}$ a database  of size $N$. 
    \begin{itemize}
        \item If $Q$ is $q$-hierarchical, then it can be maintained with $O(N)$ preprocessing time, constant time per single-tuple update and constant enumeration delay.
        \item Otherwise, there is no algorithm that, for any $\gamma > 0$, maintains $Q$ with arbitrary preprocessing time, $O(N^{1/2-\gamma})$ update time, and $O(N^{1/2-\gamma})$ enumeration delay, unless the OuMv conjecture\footnote{The original statement uses the Online Matrix-Vector Multiplication (OMv) conjecture, which implies the OuMv conjecture~\cite{Henzinger:OMv:2015}.} fails.
    \end{itemize}
\end{thm}

The $q$-hierarchical queries form a strict subclass of the class of free-connex $\alpha$-acyclic queries. Their following syntactic definition allows to check whether a query is $q$-hierarchical in polynomial time in the size of the query.

\begin{definition}[\cite{2011Suciu,BerkholzKS17}]\em
A conjunctive query is {\em hierarchical} if for any two variables $X$ and $Y$, it holds $\mathit{atoms}(X) \subseteq \mathit{atoms}(Y)$, $\mathit{atoms}(Y) \subseteq \mathit{atoms}(X)$, or $\mathit{atoms}(X) \cap \mathit{atoms}(Y) = \emptyset$. 

A hierarchical query is {\em $q$-hierarchical} if for any two variables $X$ and $Y$ with 
$\mathit{atoms}(X) \supset \mathit{atoms}(Y)$, it holds: if $Y$ is free, then $X$ is free.     
\end{definition}

\begin{ex}\label{ex:q-hierarchical}\em 
    A simple non-hierarchical query is: 
    \begin{align*}
        Q = \sum_{X,Y} R(X)\cdot S(X,Y)\cdot T(Y).
    \end{align*}
    The query becomes hierarchical if we drop any of the atoms in $Q$. A simple hierarchical but not $q$-hierarchical query is: 
    \begin{align*}
        Q(X) = \sum_{Y} R(X,Y)\cdot S(Y).
    \end{align*}\end{ex}

To achieve the desired maintenance complexity, a new technique is needed: factorized representation~\cite{Olteanu:ICDT:12, OlteanuZ15}. This is a lossless compression of the query output that allows efficient maintenance and constant delay enumeration of the tuples in the query output. 

\begin{ex}\em 
    We explain the role of factorization for the maintenance of the $q$-hierarchical query $Q(Y,X,Z) = R(Y,X)\cdot S(Y,Z)$, where the relations have size $O(N)$. Consider a single-tuple update $\delta R=\{(y,x)\mapsto m\}$. This may trigger $O(N)$ changes to the query output in case $y$ is paired with $O(N)$ $Z$-values in $S$. 

    Consider the factorization of the query output depicted by the hierarchy of views in Fig.~\ref{fig:q-hier-view-trees} (left): for each $Y$-value $y$ in both $R$ and $S$, we store separately the $X$-values paired with $y$ in $R$ from the $Z$-values paired with $y$ in $S$. This ensemble of views is a view tree.
    
    We can use the view tree to reduce the update time from linear to constant. To process an update $\delta R$, we transform the view tree into the delta view tree from Fig.~\ref{fig:q-hier-view-trees}(middle). We propagate $\delta R$ bottom-up in this delta view tree and compute the delta views along the path to the root. Each of these delta views can be computed in constant time: $\delta V_R$ projects away $x$ from $\delta R$ and $\delta V_{RS}$ requires one lookup into $V_S$ with the key $y$. Similarly, we obtain constant update time for a single-tuple update to $S$ using the delta view tree from Fig.~\ref{fig:q-hier-view-trees}(right).
    Each delta view is then added to its corresponding view. 

    The query output is not contained in one view, instead it is distributed over the views. The view tree is calibrated top-down by its construction and maintenance: The $Y$-values in $V_{RS}$ are in both $R$ and $S$.
    To enumerate the tuples in the query output, we proceed as follows. For each $Y$-value in $V_{RS}$, we look up into $R(y,X)$ and $S(y,Z)$ to enumerate all possible pairs $(X,Z)$ paired with $y$. The payload of an output tuple $(y,x,z)$ is the product of the payloads of $R(y,x)$ and $S(y,z)$. Positioning the iterators at the beginning of each $y$-block in $R$ and $S$, as well as at the beginning of $V_{RS}$ takes constant time and each lookup in these views also takes constant time. We therefore achieve constant enumeration delay.
\end{ex}

\begin{figure}[t]
  \centering   
  \includegraphics[width=0.45\textwidth]{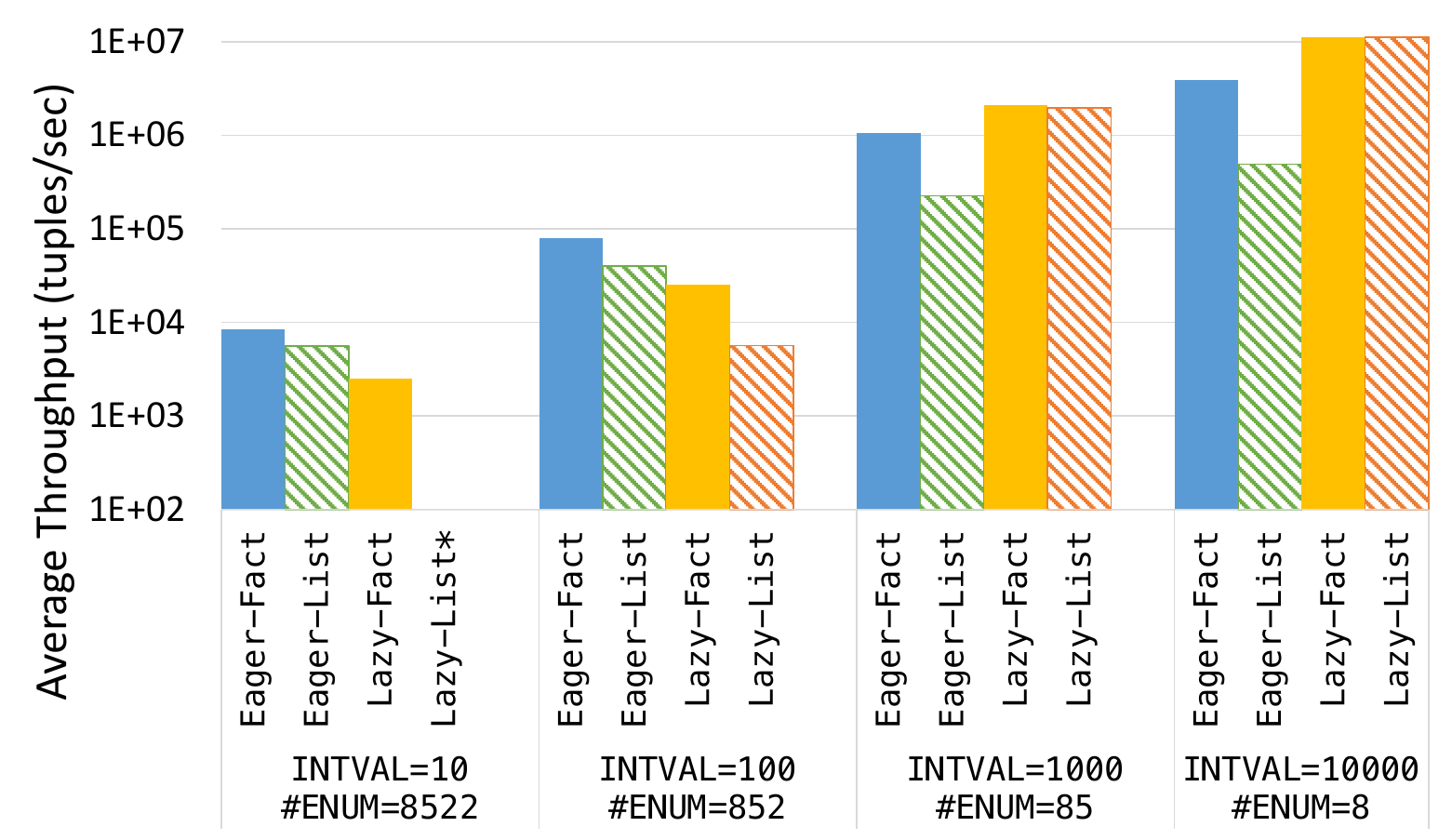}
  \caption{
  \cite{FIVM:VLDBJ:2023}
  Incremental maintenance of the result of a $q$-hierarchical query over the   {\em Retailer} dataset under update batches and requests to enumerate all tuples in the query output after every \texttt{INTVAL} update batches; a batch has 1000 single-tuple inserts; \texttt{\#ENUM} denotes the overall number of the enumeration requests. The throughput is not shown for Lazy-list* for \texttt{INTVAL}=10 as it did not finish within the 50 hours.
  }
  \label{fig:q-hier-experiment}
\end{figure}

Given a $q$-hierarchical query $Q$ and a database $\calD$, we can construct in time polynomial in the size of $Q$ and linear in the size of $\calD$ a view tree that supports the best possible maintenance for $Q$ over updates to $\calD$~\cite{DBLP:conf/sigmod/NikolicO18, FIVM:VLDBJ:2023}. 
Maintenance using such view trees can be arbitrarily faster than using other view trees or even other approaches. Dynamic Yannakakis~\cite{Idris:dynamic:SIGMOD:2017}, F-IVM~\cite{DBLP:conf/sigmod/NikolicO18, FIVM:VLDBJ:2023},  Crown~\cite{DBLP:journals/pvldb/WangHDY23}, and the RelationalAI engine provide varying support for $q$-hierarchical queries. They all show significant speedups over classical IVM approaches.

Fig.~\ref{fig:q-hier-experiment} shows the throughput of four IVM approaches for a $q$-hierarchical query that joins five relations: eager-list (DBToaster \cite{DBT:VLDBJ:2014}), eager-fact (F-IVM~\cite{FIVM:VLDBJ:2023}), lazy-list (delta queries), and lazy-fact (a hybrid of F-IVM and delta queries). They are defined on two dimensions: (1) we either propagate the updates as far as possible in the view tree (eager) or we only update the input relations and then construct the query output on an enumeration request (lazy); (2) the query output is either maintained as a list of tuples (list) in the payload at the top view or as a factorized representation in the payload of several views (fact). They all use the same view tree, which is their best view tree. The factorization approaches outperform the others in all but the trivial cases of few enumeration requests, where the representation of the query output plays no role. This is as expected, since the enumeration delay and the update time can both remain constant for $q$-hierarchical queries only if the query output is kept factorized over the views in the view tree.

The quest for the best possible maintenance goes beyond the $q$-hierarchical queries. Many queries are not hierarchical, e.g., a path join with more than two relations, yet they may enjoy the best possible maintenance due to: the structure of the query workload, data integrity constraints, or restrictions of the update batch.

\subsection{Cascading Q-hierarchical Queries}
    
When maintaining a set of queries, there are opportunities to reuse the maintenance of a query for the maintenance of other queries. In particular, we can piggyback the maintenance of a non-$q$-hierarchical query $Q_1$ on the maintenance of a $q$-hierarchical query $Q_2$ so that the set of two queries can be maintained with constant update time and delay. This requires a $q$-hierarchical rewriting of $Q_1$ using $Q_2$~\cite{DBLP:journals/lmcs/GeckKSS23}. We achieve (amortized) constant update time and constant enumeration delay for both $Q_1$ and $Q_2$ under two conditions: (i) we enumerate the outputs of both $Q_1$ and $Q_2$; and (ii) the enumeration for $Q_2$ is triggered before the enumeration for $Q_1$.

\begin{ex}\em 
Consider the following two queries:
\begin{align*}
    Q_1(A,B,C,D) &= R(A,B) \cdot S(B,C) \cdot T(C,D)\\
    Q_2(A,B,C)   &= R(A,B) \cdot S(B,C)
\end{align*}
Query $Q_2$ is $q$-hierarchical, while $Q_1$ is not. That is, when taken in isolation, $Q_2$ can be maintained with constant update time and delay, while $Q_1$ cannot. Yet we can rewrite $Q_1$ to use $Q_2$, since there is a trivial homomorphism (the identity function) from $Q_2$ to $Q_1$:
\begin{align*}
    Q'_1(A,B,C,D) &= Q_2(A,B,C) \cdot T(C,D)
\end{align*}
The rewriting $Q'_1$ is equivalent to $Q_1$ and $q$-hierarchical.

Fig.~\ref{fig:cascading-example} shows a view tree used for the maintenance of both $Q_2$ and $Q_1$. Since $Q_2$ is $q$-hierarchical, its view tree (depicted in the bottom box) can maintain it with constant update time and enumeration delay. Updates to $R$ and $S$ are propagated up to the view $V_{RS}$, while updates to $T$ are propagated up to the view $V_{RST}$. 

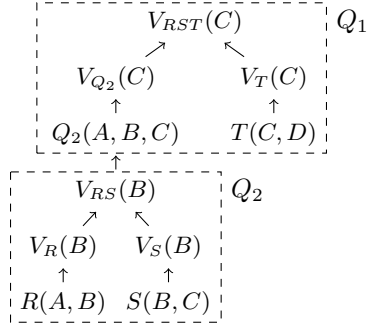
\begin{figure}
\begin{center}
\begin{tikzpicture}[xscale=0.7, yscale=0.5]

    \node at (3.5, 4.5) { $Q_1$};
    \node at (1.5, -0) { $Q_2$};

    \draw [dashed] (-2.5,1.0) rectangle (3.0,5.0);

    \node at (0.5, 4.5) (vrst) {\small $V_{RST}(C)$};
    \node at (2, 3.0) (vt) {\small $V_T(C)$} edge [->] (vrst);
    \node at (2, 1.5) (t) {\small $T(C,D)$} edge [->] (vt);
    
    \node at (-1, 3.0) (vq2) {\small $V_{Q_2}(C)$} edge[->] (vrst);
    \node at (-1, 1.5) (q2) {\small $Q_2(A,B,C)$} edge[->] (vq2);

    \draw [dashed] (-3,0.5) rectangle (1,-3.5);
    
    \node at (-1, 0) (vrs) {\small $V_{RS}(B)$} edge[->] (q2);
    \node at (-2, -1.5) (vr) {\small $V_R(B)$} edge[->] (vrs);
    \node at (0, -1.5) (vs) {\small $V_S(B)$} edge[->] (vrs);
    \node at (-2, -3.0) (r) {\small $R(A,B)$} edge[->] (vr);
    \node at (0, -3.0) (s) {\small $S(B,C)$} edge[->] (vs);
    
\end{tikzpicture}
\end{center}
    \caption{View tree for the maintenance of two cascading $q$-hierarchical queries: The non-hierarchical query $Q_1$ uses the maintenance of the $q$-hierarchical query $Q_2$. }
    \label{fig:cascading-example}
\end{figure}

An enumeration request for $Q_2$ is resolved as for $q$-hierarchical queries. Yet an enumeration request for $Q_1$ cannot be resolved before enumerating the output of $Q_2$! The reason is that we use the enumeration of the output tuples of $Q_2$ to piggyback (or: asymptotically cover the cost of) the propagation of such tuples up the view tree through the view $V_{Q_2}$ to the top view $V_{RST}$. There is a constant overhead for each such propagation. Once the enumeration for $Q_2$ has finished, the propagation of $Q_2$'s tuples have updated the views $V_{Q_2}$ and $V_{RST}$ and we can enumerate the distinct tuples in the output of $Q_1$ at any time.
Experiments~\cite{cavier_thesis} show that $Q'_1$ can have a higher throughput than $Q_1$, assuming we update both queries and enumerate the output of $Q_2$ before the output of $Q'_1$.
\end{ex}

A treatment of theoretical and practical aspects of cascading $q$-hierarchical queries is given in recent work~\cite{cavier_thesis}.

\subsection{Queries with Free Access Patterns}

Access to data may come with restrictions. For instance, to access the flights from a flight booking database behind a web interface, one has to specify the date, departure, and destination. Built-in predicates are also a common source of access restriction: For instance, given a user-defined function $f(a,b,c)$ that maps $(a,b)$ to $c$, an output value for $c$ can be obtained when input values for $a$ and $b$ are given. Such access restrictions can be made explicit in the query by means of a separation of the free variables into input and output: The query returns tuples over the output variables given a tuple of values over the input variables. Recent work~\cite{DBLP:conf/icdt/00020OZ23,kara2024conjunctive} considers the problem of maintaining such conjunctive queries with free access patterns (CQAPs) and introduces the class of tractable CQAPs, which are precisely those CQAPs that admit constant update time and constant enumeration delay. The $q$-hierarchical queries are the class of tractable CQAPs without input variables. Yet there are tractable CQAPs that are not $q$-hierarchical and not even acyclic. 

\begin{ex}[\cite{kara2024conjunctive}]\em
The following triangle detection problem is a tractable CQAP:
 Given three nodes in a graph, do they form a triangle? 
 This problem can be expressed as:  
 \begin{align*}
 Q(\cdot | A,B,C) = E(A,B) \cdot E(B,C) \cdot E(C,A), 
 \end{align*}
 where $E$ is the edge relation of the graph, The variables in the query head that appear after the symbol ``$|$'', i.e., $A,B,C$, are input variables. This query has no output variable (expressed by $\cdot$). 
 The following edge triangle listing problem is not a tractable CQAP (albeit it can be maintained optimally): Given an edge in a graph, the task is to list all triangles containing this edge: 
\begin{align*}
    Q(C | A,B) = E(A,B) \cdot E(B,C) \cdot  E(C,A).
\end{align*}
 A further example of a tractable CQAP is: 
 \begin{align*}
    Q(A | B) = S(A, B) \cdot T (B),
 \end{align*}
 which outputs $A$-values for any given input $B$-value.
\end{ex}

All tractable CQAPs can be characterized syntactically much in the spirit of the $q$-hierarchical characterization.
\begin{definition}[\cite{DBLP:conf/icdt/00020OZ23}]\em
    The {\em fracture} of a CQAP $Q$ is a CQAP $Q_\dagger$ constructed as follows. We start with $Q_\dagger$ as a copy of $Q$. We replace each occurrence of an input variable by a fresh variable. Then, we compute the connected components of the modified query. Finally, we replace in each connected component of the modified query all new variables originating from the same input variable by one fresh input variable.

    Given two variables $A$ and $B$ in $Q$, we say that $B$ {\em dominates} $A$ if $\mathit{atoms}(A) \subset \mathit{atoms}(B)$. The query $Q$ is {\em free-dominant} \footnote{The properties $q$ and free-dominant are the same.} ({\em input-dominant}) if for any two variables $A$ and $B$, it holds: if $A$ is free (input) and $B$ dominates $A$, then $B$ is free (input).

    The CQAP $Q$ is {\em tractable} if its fracture is hierarchical, free-dominant, and input-dominant.
\end{definition}

We obtain a dichotomy for CQAPs like for $q$-hierarchical queries:

\begin{thm}[\cite{DBLP:conf/icdt/00020OZ23}]\label{thm:cqap}
	Let any CQAP query $Q$ and database $\mathcal{D}$.
	\begin{itemize}
	\item If $Q$ is tractable, then it can be maintained with   
	$O(|\mathcal{D}|)$ preprocessing time,
	constant time per single-tuple update and constant enumeration delay.
	\item If $Q$ is not tractable and has no repeating relation symbols,
		then there is no  algorithm that maintains $Q$ with arbitrary preprocessing  time, 
		$\bigO{|\mathcal{D}|^{\frac{1}{2} - \gamma}}$ 
		 update time and  enumeration delay,
		for any $\gamma >0$, unless the OuMv conjecture fails. 
	\end{itemize}
\end{thm}

\subsection{Data Integrity Constraints}

In practice, databases satisfy integrity constraints. An interesting observation is that non-hierarchical queries may behave like hierarchical queries when computed over databases satisfying integrity constraints. We discuss next two such cases: functional dependencies and primary key - foreign key constraints.

\paragraph{Functional Dependencies}
A study\footnote{\url{http://www.cs.ox.ac.uk/people/dan.olteanu/papers/icde09queries.html}} of the 22 TPC-H queries shows that eight Boolean versions and 13 non-Boolean versions of these queries are hierarchical. When considering the functional dependencies that hold on the TPC-H dataset, four more Boolean and four more non-Boolean queries become hierarchical~\cite{DBLP:conf/icde/OlteanuHK09}. Also in real-world projects, functional dependencies can turn a large fraction of the query workload into hierarchical queries: e.g., 76\% of approximately 6000 queries become $q$-hierarchical under functional dependencies in one such project at RelationalAI. Functional dependencies also appear due to built-in predicates, e.g., we can derive from $A+B=C$ that any pair of variables implies the third variable.

\nop{The following is adapted from prior work~\cite{DBLP:conf/icde/OlteanuHK09,FIVM:VLDBJ:2023}.}
Given a set $\Sigma$ of functional dependencies (fds), let $\textsf{C}_\Sigma(S)$ denote the closure of the set $S$ of variables under $\Sigma$. For instance, given the fds $\Sigma=\{A\rightarrow C;BC \rightarrow D\}$, it holds $\textsf{C}_\Sigma(\{A,B\})=\{A,B,C,D\}$.

\begin{definition}[\cite{DBLP:conf/icde/OlteanuHK09}]\em
  Given a set $\Sigma$ of fds and a query 
\begin{align*}
  Q(X_1,\ldots,X_f) = \sum_{X_{f+1}} \ldots \sum_{X_{m}} \prod_{i \in [n]} R_i(\mathcal{S}_i),  
\end{align*}
the \emph{$\Sigma$-reduct} of $Q$ under $\Sigma$ is:
\begin{align*}
  Q(\textsf{C}_\Sigma(X_1,\ldots,X_f)) = \sum_{X_{f+1}} \ldots \sum_{X_{m}} \prod_{i \in [n]} R_i(\textsf{C}_\Sigma(\mathcal{S}_i)),  
\end{align*}
\end{definition}
The $\Sigma$-reduct is an extension of the query such that the schema of each relation consists of all variables in its closure under $\Sigma$.  

\begin{ex}\label{ex:hierarchical-fds}\em
The query used in the experiment in Fig.~\ref{fig:q-hier-experiment} is not $q$-hierarchical, but becomes $q$-hierarchical in the presence of fds. It is the join of five relations in a real-world Retailer dataset. We show below four of them and replace the non-join variables by dots:

\begin{tabular}{llllll}
\texttt{Inventory}: & locn, & dateid, & ksn, &       & $\ldots$ \\
\texttt{Weather}:   & locn, & dateid, &      &       & $\ldots$ \\
\texttt{Location}:  & locn, &         &      & zip,  & $\ldots$\\
\texttt{Census}:    &       &         &      & zip,  & $\ldots$ \\
\end{tabular}

This query is not hierarchical, since the sets of atoms of the variables zip and locn intersect, but one set is not included in the other. Thanks to the fd zip $\rightarrow$ locn, the last relation becomes Census(locn, zip, $\ldots$), so the reduct of the join query is hierarchical.
\end{ex}

When this rewriting using fds turns a non-$q$-hierarchical query into a $q$-hierarchical $\Sigma$-reduct, we can use the latter to derive a maintenance mechanism for the former. In particular, we use the view tree for the $q$-hierarchical $\Sigma$-reduct to construct a view tree for the original query that ensures the desirable best maintenance.

\begin{thm}[\cite{FIVM:VLDBJ:2023}]
  Let a query $Q$ and a database $\mathcal{D}$ that satisfies a set $\Sigma$ of functional dependencies. If the  $\Sigma$-reduct of $Q$ is $q$-hierarchical, then $Q$ can be maintained with $O(|\mathcal{D}|)$ preprocessing, constant time per single-tuple update and constant enumeration delay.
\end{thm}

\begin{ex}[\cite{FIVM:VLDBJ:2023}]\label{ex:fds}\em
    (This example is inspired from a RelationalAI workload, where we changed the names of relations and variables.)
    We would like to maintain the non-hierarchical query:
    \begin{align*}
        Q(Z,Y,X,W) =  R(X,W) \cdot S(X,Y) \cdot  T(Y,Z)
    \end{align*}
    over a database that satisfies the fds $\Sigma=\{X\rightarrow Y, Y\rightarrow Z\}$. 
    
    The $\Sigma$-reduct of $Q$ is $q$-hierarchical:
    \begin{align*}
      Q'(Z,Y,X,W) = R'(X,Y,Z,W) \cdot S'(X,Y,Z) \cdot T'(Y,Z)
    \end{align*}
    Fig.~\ref{fig:q-hier-fd-view-trees} shows a view tree for $Q$ (the view for $Q'$ is the same, the only change is in the signature of the relations at the leaves). It also shows the delta view tree for an update $\delta R$. We can achieve constant time per single-tuple update using this maintenance strategy as follows. The view $\delta V_W$ projects away $w$ from $\delta R$. To compute $\delta V'_X$, we look up into $S$ to fetch all $Y$-values paired with value $x$. Since $X\rightarrow Y$, we know there is at most one such $Y$-value, call it $y$. The delta view $\delta V_X$ projects away $x$ from $\delta V'_X$. To compute $\delta V'_Z$, we look up into $T$ to fetch all $Z$-values paired with $y$. This time, the second functional dependency $Y\rightarrow Z$ comes to the rescue: There is at most one $Z$-value, call it $z$, paired with $y$ in $T$. Each delta view thus has constant size and takes constant time to compute and add to their corresponding views.
    To enumerate the distinct tuples $(z,y,x,w)$ in the query output, we use again the observation that the view tree is calibrated top-down by construction, so each tuple in a view joins with at least one tuple in the views below in the view tree. This ensures constant enumeration delay.
\end{ex}

\begin{figure}
    \centering

    \begin{tikzpicture}[xscale=0.8, yscale=0.6]
        \node at (-2.0, -4.0) (vyz) {\small $V'_Z(Z,Y)$};
        \node at (-2.0, -6.0) (vy) {\small $V_X(Y)$} edge[-] (vyz);
        \node at (-2, -8.0) (vxy) {\small $V'_X(Y,X)$} edge[-] (vy);
        \node at (-2, -10.0) (vx) {\small $V_W(X)$} edge[-] (vxy);
        \node at (-2.0, -12.0) (f) {\small  $R(X,W)$} edge[-] (vx);
        \node at (0.25, -12.0) (g) {\small $S(X,Y)$} edge[-] (vxy);
        \node at (0.25, -6.0) (h) {\small  $T(Y,Z)$} edge[-] (vyz);
    \end{tikzpicture}
      \begin{tikzpicture}[xscale=0.8, yscale=0.6]
            \node at (-2.0, -4.0) (vyz) {\small ${\color{blue}\delta V'_Z(z,y)}$};

            \node at (0.5, -4.0) {\small FD: $Y \rightarrow Z$};

            \node at (0.25, -6.0) (h) {\small  $T(Y,Z)$} edge[-] (vyz);

            \node at (-2.0, -6.0) (vy) {\small ${\color{blue}\delta V_X(y)}$} edge[-] (vyz);

            \node at (-2, -8.0) (vxy) {\small ${\color{blue}\delta V'_X(y,x)}$} edge[-] (vy);
            \node at (0.5, -8.0) {\small FD: $X \rightarrow Y$};

            \node at (-2, -10.0) (vx) {\small ${\color{blue}\delta V_W(x)}$} edge[-] (vxy);

            \node at (-2.0, -12.0) (f) {\small  ${\color{blue}\delta R(x,w)}$} edge[-] (vx);

            \node at (0.25, -12.0) (g) {\small $S(X,Y)$} edge[-] (vxy);
      \end{tikzpicture}
    
        \caption{IVM for the non-hierarchical query in Ex.~\ref{ex:fds} in the presence of functional dependencies. Left: View tree for the query. Right: Delta view tree for update $\delta R$ to relation $R$.}
    \label{fig:q-hier-fd-view-trees}
\end{figure}
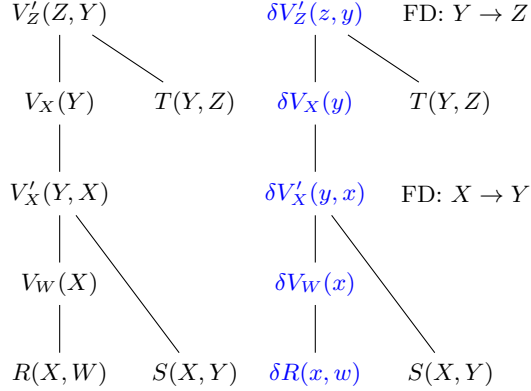

The above treatment of functional dependencies can be generalized to bounded degree  constraints: whereas $X\rightarrow Y$ states that an $X$-value $x$ can only be paired with one $Y$-value $y$ in $S$, the bounded degree constraint states that $x$ can only be paired with a constant number of distinct $Y$-values in $S$. This guarantees a constant-time lookup into $S$ for each update to $R$.

The small domain constraints, which state that a column has a constant number of values, are also studied in the IVM literature~\cite{BerkholzKS18}.

\paragraph{Primary Key - Foreign Key Constraints}
Whenever a relation has a primary key that is used as a foreign key in another relation, a {\em consistent} database  satisfies the constraint that a value for the foreign key must also exist as a value for the primary key. We call a batch of updates {\em valid} if it maps a consistent database to another consistent database. Our key observation is that there are non-hierarchical queries that can be maintained with (amortized) constant update time and enumeration delay under valid update batches. This holds {\em regardless of the execution order} of the updates. Typically, there are triggers that reject an update that can leave the database in an inconsistent state: they do not delete a value for a primary key if it is still used as foreign key and do not insert a value as a foreign key if it is not already a value for a primary key. Yet for optimization purposes, e.g., out-of-order execution of updates, it is desirable to allow intermediate inconsistent databases as long as the final (committed) database is consistent.

\begin{ex}\em 
    Consider a query over the simplified IMDB database used in the JOB benchmark~\cite{DBLP:journals/vldb/LeisRGMBKN18}. The query joins the following relations, where the non-join variables are denoted by dots:    

\hspace*{-1em}\begin{tabular}{llll}
\texttt{Title (T)}: & movie\_id,  &       & $\ldots$ \\
\texttt{Movie\_Companies (M)}:   & movie\_id, & company\_id,  &   $\ldots$ \\
\texttt{Company\_Name (C)}:  &   & company\_id,      & $\ldots$\\
\end{tabular}

The variables movie\_id and company\_id are foreign keys in \texttt{M} and primary keys in \texttt{T} and respectively \texttt{C}. The query is not $q$-hierarchical, nor its reduct under the fds defined by the primary keys. Consider a batch of updates to the relations. An insert of a new company with id $c$ into \texttt{C} takes constant time, if there are no records for $c$ in \texttt{M}. Yet if there are $n$ records for $c$ in \texttt{M} before the insert of $c$ into \texttt{C}, then it takes $O(n)$ to process this insert. It also means that before this insert (1) the database was  inconsistent and (2) each of the $n$ inserts into \texttt{M} took constant time, as there was no joining partner in \texttt{C}. On average, each of the $n$ inserts into \texttt{M} and the one insert into \texttt{C} needs amortized constant time. After the insert into \texttt{C}, the inconsistency due to the absence of the primary key value $c$ is resolved.
Similar reasoning applies to a delete of the key $c$ from \texttt{C}: If there are no movies paired with $c$ in \texttt{M}, then this delete takes constant time and the database is consistent with respect to the key $c$. If there are $n$ movies paired with $c$ in \texttt{M}, then this delete takes $O(n)$ time. Yet in order for the database to become consistent with respect to $c$, each of these $n$ records in \texttt{M} have to be eventually deleted, each taking constant time. On average, each of the $n+1$ deletes will thus take constant time.
\end{ex}


\subsection{Static versus Dynamic Relations}

In practice, there are relations, like Inventory and Sales, that are updated much more frequently than others, like Stores and Demographics. For the specific task of maintaining the output of a query over such relations within a given time frame (e.g., a day or even a month), we may as well practically consider the latter relations static. This can have a significant positive impact on the complexity of query maintenance: Non-$q$-hierarchical queries may enjoy constant update time and constant enumeration delay under the realistic restriction that some relations are static. Recent work~\cite{StaticDynamic:2024} puts forward a syntactic characterization of those tractable queries that admit the best possible maintenance in this mixed static-dynamic setting. The class of such queries strictly contains the class of $q$-hierarchical queries, which admit the best possible maintenance already in the all-dynamic setting. Intuitively, for such queries we can construct view trees much like for $q$-hierarchical queries, where now the materialized views are dynamic, if they are defined over at least one dynamic relation, and static, if they are defined over static relations only. To ensure constant update time, the schema of each dynamic view covers the schema of each of its sibling views; in the all-dynamic setting, this property holds for all views since they are all dynamic, whereas now it only needs to hold for the dynamic views. To ensure constant enumeration delay, the set of free variables of the query is precisely the set of variables of all the views in a subtree containing the root of the view tree; this property is the same as in the all-dynamic setting. 

\begin{ex}\em
 Consider the following query, where the dynamic relations $R$ and $S$ are adorned with the superscript $d$, while the static relation $T$ is adorned with $s$:
\begin{align*}
    Q(A,B,C) = \sum_D R^d(A,D) \cdot S^d(A,B) \cdot T^s(B,C)
\end{align*}
This query is not $q$-hierarchical. yet we can maintain it efficiently using the following tree of materialized views:
\begin{center}
\begin{tikzpicture}[xscale=1.6, yscale=0.9]
    \node at (-2, 2) (vrst) {\small $V^d_{RST}(A) = V^d_B(A) \cdot V^d_D(A)$};
    \node at (-3, 1) (vd) {\small $V^d_D(A) = \sum_D R^d(A,D)$} edge[->] (vrst);
    \node at (-3, 0) (r) {\small $R^d(A,D)$} edge[->] (vd);
    \node at (-1, 1) (vb) {\small $V^d_B(A) = \sum_B V^d_S(A,B)$} edge[->] (vrst);
    \node at (-1, 0) (vrs) {\small $V^d_S(A,B) = S^d(A,B) \cdot V^s_C(B)$} edge[->] (vb);
    \node at (-2, -1) (vr) {\small $S^d(A,B)$} edge[->] (vrs);
    \node at (0, -1) (vs) {\small $V^s_C(B) = \sum_C T^s(B,C)$} edge[->] (vrs);
    \node at (0, -2.0) (s) {\small $T^s(B,C)$} edge[->] (vs);
\end{tikzpicture}
\end{center}
\nop{
\begin{align*}
    V^s_C(B) &= \sum_C T^s(B,C)\\
    V^d_S(A,B) &= S^d(A,B) \cdot V^s_C(B) \\
    V^d_B(A) &= \sum_B V^d_S(A,B) \\
    V^d_D(A) &= \sum_D R^d(A,D) \\
    V^d_{RST}(A) &=  V^d_B(A) \cdot V^d_D(A)
\end{align*}
}
Single-tuple updates to both $R$ and $S$ require constant time; if we would allow updates to $T$ as well, then one such update would take linear time, as there may be linearly many $A$-values paired with a $B$-value in $S$. To enumerate, we proceed top-down: For each $A$-value $a$ in $V^d_{RST}$, we iterate over each $B$-value $b$ paired with $a$ in  $V^d_S(A,B)$, and then iterate over the $C$-values paired with $b$  in $T$.

The following non-$q$-hierarchical query can also be maintained efficiently, albeit after quadratic time preprocessing needed to join the static relations $S$ and $T$ on the bound variable $B$:
\begin{align*}
  Q(A, C, D) = \sum_B R^d(A, D) \cdot S^s(A, B)\cdot T^s(B, C)\cdot U^d(D)  
\end{align*}

Even the non-hierarchical query from Ex.~\ref{ex:q-hierarchical} can be maintained efficiently as long as the in-between relation $S$ is static:
\begin{align*}
    Q(A, B) = R^d(A) \cdot S^s(A, B)\cdot T^d(B)
\end{align*}
The reason is that the query output can only be a subset of the static relation $S$. The  dynamic relations $R$ and $T$ act as indices that identify one of the exponentially many possible outputs from the powerset of $S$. This requires however exponential time preprocessing.
\end{ex}

\subsection{Insert-only versus Insert-delete}

The structure of the update batch can also influence the update time. Insert-only batches can be processed asymptotically faster than batches of both inserts and deletes~\cite{Kara:TODS:2020,DBLP:journals/pvldb/WangHDY23,DBLP:journals/corr/abs-2312-09331}. For instance, every $\alpha$-acyclic join query $Q$ can be maintained with amortized constant time for a single-tuple insert and constant enumeration delay. A simplified argument is fairly direct\footnote{To keep it simple, this argument disregards a $\log n$ factor in the time complexities.}: Given a batch of $n$ inserts to the initially empty database, we can first update the input database and then the factorized output of $Q$ over the updated database in $O(n)$ time~\cite{OlteanuZ15}. The average time per insert is thus constant. 
Recall from Theorem~\ref{th:q-hierarchical} that if we were to support both inserts and deletes, then the non-$q$-hierarchical queries that are $\alpha$-acyclic cannot enjoy constant update time and constant delay.
In the insert-only setting, off-the-shelf worst-case optimal join algorithms, such as LeapFrog TrieJoin~\cite{LeapFrogTrieJoin2014}, combined with delta queries can be used to achieve the best known update time for arbitrary join queries, while supporting constant (full or delta) enumeration delay~\cite{DBLP:journals/corr/abs-2312-09331}. 

In the insert-delete setting, the update time can be determined by the refined notion of enclosureness of the update sequence, which observes the insert/delete timestamps of the tuples~\cite{DBLP:journals/pvldb/WangHDY23}.

\begin{figure}

\centering
    
\tdplotsetmaincoords{73}{170}
\begin{tikzpicture}[xscale=1.92, yscale=1.44,tdplot_main_coords]
  \coordinate (O) at (0,0,0);


  \draw[thick,dotted,color=gray] (0,0,0) -- (0.5,0,0);
  \draw[thick,->] (0.5,0,0) -- (2,0,0) node[anchor=north]{\footnotesize $\log_{N}${delay}};
  \draw[thick,dotted,color=gray] (0,0,0) -- (0,3,0);
  \draw[thick,->] (0,3,0) -- (0,3.5,0) node[anchor=north]{\footnotesize $\log_{N}${preprocessing time}};
  \draw[thick,dotted,color=gray] (0,0,0) -- (0,0,0.5);
  \draw[thick,->] (0,0,0.5) -- (0,0,1.5) node[anchor=south]{\footnotesize $\log_{N}${update time}};

  \node[color=gray] at (-0.1,-0.1,0) () {\footnotesize $0$};
  \coordinate (P1) at (1,1,0);
  \node[] at (1,1,-0.2) () {\footnotesize $(1,0,1)$};
  \coordinate (P1x) at (1,0,0);
  \node[] at (1,-0.6,0.0) () {\footnotesize $1$};

  \coordinate (P1y) at (0,1,0);

  \coordinate (P2) at (0,1,1);
  \coordinate (P2xy) at (0,1,0);
  \coordinate (P2y) at (0,1,0);
  
  \coordinate (P2z) at (0,0,1);

  \coordinate (P3) at (0,2,2);
  \coordinate (P3z) at (0,0,2);


  \draw[dashed, color=black] (P1x) -- (P1);
  \draw[dashed, color=black] (P1y) -- (P1);

  \node[color=gray] at (-0.16,0.7,0) () {\footnotesize $1$};


  \node[] at (0.25,0,1) () {\footnotesize $\delta = 1$};

  \draw[dashed, color=black] (P2) -- (P2y);
  \draw[dashed, color=black] (P2) -- (P2z);

  \filldraw[blue] (0,1,1) ellipse(1.5pt and 2pt);
  \node at (-0.23,1,1) {\color{blue} \footnotesize Eager};

  \filldraw[goodgreen] (1,1,0) ellipse(1.5pt and 2pt);
  \node at (1.2,1,0) {\color{goodgreen} \footnotesize Lazy};

  \node at (0,1,0.75) {\color{red} \footnotesize \ivme};

  \coordinate (X_half) at (0.5,0,0);
  \coordinate (Z_half) at (0,0,0.5);
  \coordinate (XZ_half) at (0.5,0,0.5);
  \fill[gray!50, opacity=0.4] (XZ_half) -- (Z_half) -- (0,3.0,0.5) -- (0.5,3.0,0.5) -- cycle;
  \fill[gray!50, opacity=0.4] (XZ_half) -- (X_half) -- (0.5,3.0,0) -- (0.5,3.0,0.5) -- cycle;
  \fill[gray!50, opacity=0.4] (0.5,3.0,0.5) -- (0.5,3.0,0) -- (0,3.0,0) -- (0,3.0,0.5) -- cycle;
  \draw[color=gray] (0.5,0,0) -- (0.5,3.0,0);
  \draw[color=gray] (0.5,0,0) -- (0.5,0,0.5);
  \draw[color=gray] (0.5,0,0.5) -- (0,0,0.5);
  \draw[color=gray] (0.5,0,0.5) -- (0.5,3.0,0.5);
  \draw[color=gray] (0.5,3.0,0.5) -- (0.5,3.0,0);
  \draw[color=gray] (0.5,3.0,0.5) -- (0,3.0,0.5);
  \draw[color=gray] (0.5,3.0,0) -- (0,3.0,0);
  \draw[color=gray] (0,3.0,0.5) -- (0,0,0.5);
  \draw[color=gray] (0,3.0,0) -- (0,3.0,0.5);

  \node at (0.55,-0.5,0) () {\footnotesize$1/2$};
  \node at (0.15,0,0.6) () {\footnotesize$1/2$};

    \filldraw[red] (0,1,1) ellipse(0.8pt and 1.04pt);    
    \filldraw[red] (1,1,0) ellipse(0.8pt and 1.04pt);
    
    \draw[dashed, color=gray] (0.5,1.0,0.5) -- (0.5,1.0,0);
    
    \draw[color=red, very thick] (P1) -- (P2);
    \filldraw[red] (0.5,1,0.5) ellipse(1pt and 1.333pt);
    \draw[->, thick, color=black, dotted] (0.8,1.0,-0.5) -- (0.51,1.05,0.45);
    \node at (0.8,1.0,-0.65) () {\color{red}\footnotesize $(1,1/2,1/2)$};

\end{tikzpicture}

\caption{Maintenance trade-off space for the query in Ex.~\ref{ex:tradeoff}.}
\label{fig:non-q-lower-bound}
\end{figure}
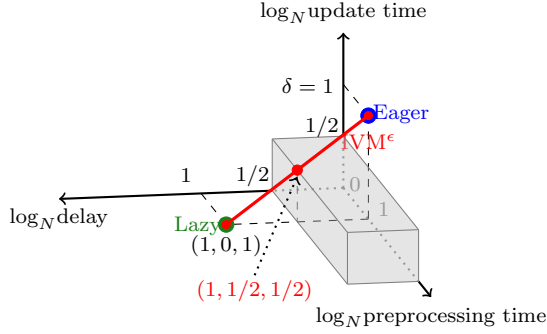

\section{Optimality Beyond the Q Continuum}
\label{sec:optimality}

Beyond the development highlighted in Sec.~\ref{sec:q}, there are results on worst-case optimal IVM for some limited classes of conjunctive queries: hierarchical queries~\cite{DBLP:conf/pods/0002NOZ20,DBLP:journals/lmcs/KaraNOZ23}  and cyclic queries (Loomis-Whitney)~\cite{DBLP:conf/icdt/KaraNNOZ19,Kara:TODS:2020} with {\em arbitrary free variables}. The single-tuple update time for such queries meets the conditional lower bound in Theorems~\ref{th:q-hierarchical} and \ref{thm:cqap}. The preprocessing time, update time, and enumeration delay exhibit a trade-off and for certain points in the trade-off space, they become (strongly or weakly) Pareto optimal.

\begin{ex}\label{ex:tradeoff}\em 
    Fig.~\ref{fig:non-q-lower-bound} depicts the trade-off between the preprocessing time, the time for a single-tuple update, and the enumeration delay for the simplest non-$q$-hierarchical query:
    \begin{align*}
        Q(A) = \sum_{B}R(A,B) \cdot S(B).
    \end{align*}
The gray cuboid is off limits, unless the Online Matrix-Vector Multiplication conjecture fails~\cite{BerkholzKS17}, since each point strictly included in the gray cuboid is a combination of some preprocessing time and $\bigO{N^{\frac{1}{2} - \gamma}}$ update time and delay for $\gamma > 0$ ($N$ is the database size).

Two extreme IVM approaches are (1) to maintain the materialized query output after every update (eager approach) and (2) to only update the input relations and materialize the query output on an enumeration request (lazy approach).
\ivme~\cite{DBLP:conf/pods/0002NOZ20} achieves each point on the line between the two extremes: it needs $O(N)$ time for preprocessing, $O(N^\epsilon)$ time for a single-tuple update, and $O(N^{1-\epsilon})$ enumeration delay. In particular, it touches the gray cuboid for $\epsilon=1/2$: At the point $(1, 1/2, 1/2)$, \ivme is weakly Pareto worst-case optimal in the sense that there can be no tighter upper bounds for both the update time and the enumeration delay (unless the OuMv conjecture fails). To achieve this trade-off, \ivme partitions relation $R$ based on heavy and light $A$-values, as shown in Sec.~\ref{sec:heavy-light}.

\nop{
The eager approach materializes the query output in $O(N)$ time. On an update $\delta R = \{(a,b) \mapsto m\}$ to $R$, it updates the query output in constant time by checking whether $b$ is in $S$. On an update $\delta S = \{(b) \mapsto m\}$ to $S$, it finds the $A$-values paired with $b$ in $R$ and updates the query output in $O(N)$ time. The query output is materialized and maintained and can thus be enumerated with constant delay.

At the other extreme, the lazy approach it may eliminate the dangling tuples from the input relations at preprocessing. On an update, it just updates the relations in constant time and defers the computation of the query to the enumeration phase. On an enumeration request, it eliminates the dangling tuples from $R$ in $\bigO{N}$ time and enumerate the result from $R$. The enumeration delay is thus $\bigO{N}$.

Our approach, \ivme~\cite{DBLP:conf/pods/0002NOZ20}, shows a trade-off between the eager and lazy approaches, which is denoted by the red line that connects the blue and green points in the figure. It achieves sub-linear update time and sub-linear enumeration delay. It is weakly Pareto worst-case optimal at point $(1.0, \frac{1}{2}, \frac{1}{2})$: there can be no tighter upper bounds for both update time and enumeration delay. 
}

\end{ex}

\section{Conclusion}
\label{sec:further}

To solve the IVM problem optimally, one can benefit from a careful analysis of: the structure of the queries when taken individually or as a set, the integrity constraints and value degrees in the data, the type of updates (insert-only versus insert-delete), and the update frequency of the relations. In practice, this analysis is an effective guide to the best possible maintenance scenario with (amortized) constant update time and enumeration delay, which translates to higher maintenance throughput than prior IVM approaches. Despite this progress, the IVM problem remains rich of challenges brought by the interplay of queries, data constraints, and updates.

\paragraph*{Further Recent Work on IVM}
Beyond conjunctive queries, there are studies on IVM for intersection joins~\cite{TaoY22}, Datalog~\cite{DBLP:journals/ai/MotikNPH19}, Differential Datalog~\cite{DBLP:conf/fossacs/AbadiMP15}, and DBSP~\cite{BudiuCMRT23}. The
descriptive complexity of recursive queries, such as reachability, under updates has also been investigated~\cite{DattaKMSZ18,DBLP:journals/sigmod/SchwentickVZ20}. There is also recent development on a distinct trend in IVM: the maintenance of complex analytics over evolving databases, which includes linear algebra computation~\cite{DBLP:conf/sigmod/NikolicEK14}, collection programming~\cite{DBLP:conf/pods/0001LT16}, in-database machine learning~\cite{DBLP:conf/sigmod/NikolicO18,DBLP:conf/sigmod/NikolicZ0O20,FIVM:VLDBJ:2023}, probabilistic reasoning~\cite{DBLP:conf/pods/BerkholzM21,kara2024conjunctive}, and Shapley-based fact attribution.

\paragraph*{Acknowledgements} The author would like to acknowledge the UZH Global Strategy and Partnerships Funding Scheme and Haozhe Zhang for his assistance.

\bibliographystyle{abbrv}
\bibliography{bibtex}

\end{document}